\documentclass[12pt,draftcls,onecolumn,english]{IEEEtran}
\IEEEoverridecommandlockouts
\makeatletter
\def\ps@IEEEtitlepagestyle{%
  \def\@oddfoot{\mycopyrightnotice}%
  \def\@evenfoot{}%
}
\def\mycopyrightnotice{%
  {\footnotesize 978-1-7281-0407-2/19/\$31.00 ©2019 IEEE\hfill}
  \gdef\mycopyrightnotice{}
}

\usepackage{cite}
\usepackage{amsmath,amsthm,amsfonts, amssymb}
\allowdisplaybreaks
\theoremstyle{plain}
\newtheorem{thm}{Theorem}
\newtheorem{lem}{Lemma}

\newtheorem{cor}{Corollary}
\newtheorem{assumption}{Assumption}
\theoremstyle{definition}
\newtheorem{defn}{Definition}

\theoremstyle{remark}
\newtheorem{rem}{Remark}

\usepackage{algorithmic}
\usepackage{graphicx}
\graphicspath{{../figure/}}
\usepackage{subfig}
\usepackage{epstopdf}

\usepackage{xtab}

\usepackage{algorithm,algorithmic}%
\usepackage{textcomp}
\usepackage{xcolor}
\def\BibTeX{{\rm B\kern-.05em{\sc i\kern-.025em b}\kern-.08em
    T\kern-.1667em\lower.7ex\hbox{E}\kern-.125emX}}
\begin{document}

\title{Chance-constrained Unit Commitment via the Scenario Approach
}

\author{Xinbo Geng, and Le Xie
\thanks{The authors are with the Department of Electrical and Computer
Engineering, Texas A\&M University, College Station, TX, 77843, USA.
(e-mail:xbgeng@tamu.edu; le.xie@tamu.edu). This work is supported in part by Electric Reliability Council of Texas, and in part by NSF Grant ECCS-1839616.}
}


\maketitle

\begin{abstract}
Keeping the balance between supply and demand is a fundamental task in power system operational planning practices. This task becomes particularly challenging due to the deepening penetration of renewable energy resources, which induces a significant amount of uncertainties. In this paper, we propose a chance-constrained Unit Commitment (c-UC) framework to tackle challenges from uncertainties of renewables. The proposed c-UC framework seeks cost-efficient scheduling of generators while ensuring operation constraints with guaranteed probability. We show that the scenario approach can be used to solve c-UC despite of the non-convexity from binary decision variables. We reveal the salient structural properties of c-UC, which could significantly reduce the sample complexity required by the scenario approach and speed up computation. Case studies are performed on a modified 118-bus system.
\end{abstract}


\section{Introduction} 
\label{sec:introduction}
Unit commitment (UC) is one of the most important decision making processes in the day-ahead operation of power systems. UC seeks the most cost-efficient on/off decisions and dispatch schedule for the generators, considering various constraints on the generators and system security under contingencies. Consideration of additional constraints such as transmission capacities in UC leads to a more general problem known as Security Constrained Unit Commitment (SCUC). The main focus of this paper is the UC problem \emph{without} transmission constraints. Possible extensions towards SCUC are discussed at the end of this paper.

UC is naturally a decision making problem with uncertainties. Traditionally, UC deals with uncertainties from unexpected events such as device failures as well as load forecast errors. Recently, the growing amount of uncertainties from renewables pose new challenges on the operations of power systems. UC, as a critical part of day-ahead scheduling, needs to be improved to consider the impacts of uncertainties.

Broadly speaking, there are two approaches for decision making under uncertainties: \emph{stochastic} optimization (SO) and \emph{robust} optimization (RO). SO relies on probabilistic models to explain
uncertainties and often optimizes the objective function in the presence of randomness. SO has found many successful applications in power systems. References \cite{takriti_stochastic_1996,wu_stochastic_2007,zheng_stochastic_2015} formulate and solve the \emph{stochastic unit commitment} problem, which typically minimizes expected commitment and dispatch costs in the presence of uncertainties. RO takes an alternative approach, in which the uncertainty model is set-based and deterministic. Recently, researchers in \cite{bertsimas_adaptive_2013} formulated and solved the \emph{robust unit commitment} problem, which minimizes the commitment and dispatch costs for the worst case in a predefined uncertainty set. 

Both approaches attract a lot of attention and are relatively successful in addressing the challenges related with uncertainties. This paper looks at the UC problem through the lens of chance-constrained optimization (CCO), which is closely related with both stochastic and robust optimization \cite{geng_data-driven_2019-2}. The main difference of CCO from SO or RO is the chance constraint (i.e. (\ref{form:cc_opt_function_cc}) and (\ref{form:cc_opt_cc})), which explicitly considers the feasibility of solutions under uncertainties. 

Various formulations of chance-constrained (security-constrained) unit commitment have been proposed, e.g. \cite{ozturk_solution_2004,pozo_chance-constrained_2013,wang_chance-constrained_2012,wang_price-based_2013,zhao_expected_2014,wu_chance-constrained_2014,tan_hybrid_2016,bagheri_data-driven_2017,zhang_chance-constrained_2017-1}. As mentioned in \cite{geng_data-driven_2019-2}, chance-constrained optimization problems can be solved via different methods. We take chance-constrained unit commitment problem as an example. It can be solved using sample average approximation \cite{wang_chance-constrained_2012,wang_price-based_2013,zhao_expected_2014,tan_hybrid_2016,bagheri_data-driven_2017,zhang_chance-constrained_2017-1} or robust optimization based techniques \cite{jiang_robust_2012}.
The scenario approach, which might be the most well-known method to solve chance-constrained optimization, was not \emph{directly} applied on the unit commitment. The only related references we found are \cite{margellos_stochastic_2013,hreinsson_stochastic_2015}, which are built upon a variation of the scenario approach \cite{margellos_road_2014}. The original scenario approach in \cite{campi_exact_2008,campi_scenario_2009} was considered \emph{not applicable} on the unit commitment problem because of the convexity assumption (see Assumption \ref{assu:convexity} in Section \ref{sub:the_scenario_approach}). This paper, however, demonstrates that the original scenario approach is indeed applicable by exploring the structure of the unit commitment problem.

The main contributions of this paper are threefold: (1) we formulate the chance-constrained unit commitment problem, and obtain the optimal solution with rigorous guarantees on the feasibility of the solution; (2) in spite of the non-convexity from commitment decisions, we show that the scenario approach is still applicable on the UC problem; (3) by exploring the structural properties of unit commitment, we greatly reduce the sample complexity required by the scenario approach.

The remainder of this paper is organized as follows. Section \ref{sec:introduction_to_chance_constrained_optimization} introduces chance-constrained optimization and the scenario approach. The deterministic and chance-constrained unit commitment problems are formulated in Section \ref{sec:unit_commitment}. 
Section \ref{sec:scenario_based_unit_commitment} applies the scenario approach on the chance-constrained unit commitment problem and analyzes its structural properties.
Numerical results are in Section \ref{sec:case_study}. Section \ref{sec:concluding_remarks} presents the concluding remarks.

The notations in this paper are standard.
All vectors are in the real field $\mathbf{R}$. We use $\mathbf{1}$ to denote an all-one vector of appropriate size. The transpose of a vector $a$ is $a^\intercal$.
The element-wise multiplication of the same-size vectors $a$ and $b$ is denoted by $a \circ b$.
Sets are in calligraphy fonts, e.g. $\mathcal{S}$. The cardinality of a set $\mathcal{S}$ is $|\mathcal{S}|$. The Cartesian product of multiple sets is denoted by $\times$, e.g. $\mathcal{U}_1 \times \mathcal{U}_2 \times \cdots \times \mathcal{U}_N$.

\section{Introduction to Chance-constrained Optimization} 
\label{sec:introduction_to_chance_constrained_optimization}
\subsection{Chance-constrained Optimization} 
\label{sub:chance_constrained_optimization}
Chance-constrained optimization is a major approach for decision making in uncertain environments.
A typical chance-constrained optimization problem is presented in (\ref{form:cc_opt_function}).
\begin{subequations}
\label{form:cc_opt_function}
\begin{align}
  \min_{x \in \mathbf{R}^n} ~ &  c^\intercal x \\
  \text{s.t.}~&  \mathbb{P}_\xi \Big( f(x,\xi) \le 0 \Big) \ge 1 - \epsilon \label{form:cc_opt_function_cc} \\
  & g(x) \le 0 \label{form:cc_opt_function_det}
\end{align}
\end{subequations}
We could write (\ref{form:cc_opt_function}) in a more compact form by defining $\mathcal{X}_\xi:= \{x \in \mathbf{R}^n: f(x,\xi) \le 0\}$ and $\chi := \{x \in \mathbf{R}^n: g(x) \le 0 \}$.
\begin{subequations}
\label{form:cc_opt}
\begin{align}
  \min_{x \in \chi}~ &  c^\intercal x  \\
  \text{s.t.}~&  \mathbb{P}_\xi \Big( x \in \mathcal{X}_\xi \Big) \ge 1 - \epsilon \label{form:cc_opt_cc}
\end{align}
\end{subequations}
Without loss of generality \cite{campi_scenario_2009}, we assume that the objective is a linear function of decision variables $x \in \mathbf{R}^n$. Variables $\xi \in \Xi$ denotes the source of uncertainties and $\Xi$ is the support of the random variable. Deterministic constraints (\ref{form:cc_opt_function_det}) are denoted by set $\chi$ in (\ref{form:cc_opt}). Constraint (\ref{form:cc_opt_function_cc}) or (\ref{form:cc_opt_cc}) is the \emph{chance constraint}. The chance constraint requires the the inner constraint $x \in \mathcal{X}_\xi$ to be satisfied with probability at least $1 - \epsilon$, where the violation probability $\epsilon$ is typically a small number (e.g. $1\%$). The set $\mathcal{X}_\xi$ depends on the realization of $\xi$ and the probability is taken with respect to $\xi$.

Since its birth in 1950s, researchers have proposed many methods to solve chance-constrained optimization problems, e.g. scenario approach, sample average approximation, and convex approximation. A detailed review and tutorial to chance-constrained optimization is in \cite{geng_data-driven_2019-2}.

\subsection{Scenario Approach} 
\label{sub:the_scenario_approach}
Scenario approach is one of the most well-known methods to solve chance-constrained optimization problems. It has been applied on various power system problems, e.g. economic dispatch \cite{modarresi_scenario-based_2018} and demand response \cite{ming_scenario-based_2017}.
The scenario approach utilizes $N$ independent and identically distributed (i.i.d.) scenarios $\mathcal{N} := \{\xi^1,\xi^2,\cdots,\xi^N\}$ to convert the chance-constrained program (\ref{form:cc_opt_function}) to the \emph{scenario problem} below:
\begin{subequations}
\label{form:scenario_problem_function}
\begin{align}
  \text{(SP)}_{\mathcal{N}}:~\min_{x \in \chi } \quad &  c^\intercal x  \\
  \text{s.t. } & f(x,\xi^i) \le 0,~ i =1,2,\cdots,N \label{form:scenario_problem_scenario_function_i}  
\end{align}
\end{subequations}
The scenario problem $\text{(SP)}_{\mathcal{N}}$ seeks the optimal solution $x_{\mathcal{N}}^*$ which is feasible for all $N$ scenarios. The scenario problem can be represented more concisely by defining $\mathcal{X}_i:= \{x \in \mathbf{R}^n: f(x, \xi^i) \le 0\}$:
\begin{subequations}
\label{form:scenario_problem_set}
\begin{align}
  \text{(SP)}_{\mathcal{N}}:~\min_{x \in \chi}~&  c^\intercal x  \\
  \text{s.t.}~&  x \in \cap_{i=1}^N \mathcal{X}_i \label{form:scenario_problem_scenario}
\end{align}
\end{subequations}
\begin{defn}[Violation Probability]
The \emph{violation probability} of a candidate solution $x^\diamond$ is defined as the probability that $x^\diamond$ is infeasible $\mathbb{V}(x^\diamond) := \mathbb{P}_\xi\big( x^\diamond \notin \mathcal{X}_\xi \big)$.
\end{defn}
The scenario approach theory aims at answering the following \emph{sample complexity} question: what is the smallest sample size $N$ such that $x_{\mathcal{N}}^*$ is feasible (i.e. $\mathbb{V}(x_\mathcal{N}^*) \le \epsilon$) to the original chance-constrained program (\ref{form:cc_opt})? Reference \cite{campi_exact_2008} provides some deep results by exploring the structural properties of the scenario problem $\text{SP}_{\mathcal{N}}$.
\begin{defn}[Support Scenario]
\label{defn:support_scenario}
A scenario $\xi^i$ is a \emph{support scenario} for the scenario problem $\text{(SP)}_{\mathcal{N}}$ if its removal changes the solution of $\text{(SP)}_{\mathcal{N}}$. $\mathcal{S}$ denotes the set of support scenarios.
\end{defn}
\begin{defn}[Non-degeneracy \cite{campi_exact_2008}] Let $x_{\mathcal{N}}^*$ and $x_{\mathcal{S}}^*$ stand for the optimal solutions to the scenario problems $\text{SP}_{\mathcal{N}}$ and $\text{SP}_{\mathcal{S}}$, respectively. The scenario problem $\text{SP}_{\mathcal{N}}$ is said to be \emph{non-degenerate}, if $c^\intercal x_{\mathcal{N}}^* = c^\intercal x_{\mathcal{S}}^*$.
\end{defn}
Theorem \ref{thm:exact_feasibility_scenario_approach} presents one of the most important results in the scenario approach theory, which is based on the non-degeneracy, feasibility and convexity assumptions below.
\begin{assumption}[Non-degeneracy\cite{campi_exact_2008,calafiore_random_2010}]
\label{ass:non-degeneracy}
For every $N$, the scenario problem $\text{SP}_{\mathcal{N}}$ is non-degenerate with probability one with respect to scenarios $\mathcal{N} = \{\xi^{1},\xi^{2},\cdots,\xi^{N}\}$.
\end{assumption}
\begin{assumption}[Feasibility and Uniqueness]
\label{asmp:feasibility_uniqueness}
Every scenario problem $\text{(SP)}_{\mathcal{N}}$ is feasible, and its feasibility region has a non-empty interior. Moreover, the optimal solution $x_{\mathcal{N}}^\ast$ of $\text{(SP)}_{\mathcal{N}}$ exists and is unique.
\end{assumption}
\begin{assumption}[Convexity]
\label{assu:convexity}
The deterministic constraint $g(x) \le 0$ is convex, and the random constraint $f(x,\xi)$ is convex in $x$ for every instance of $\xi$. In other words, the sets $\chi$ and $\mathcal{X}_i$s are convex.
\end{assumption}
\begin{thm}[\cite{campi_exact_2008,calafiore_random_2010}]
\label{thm:exact_feasibility_scenario_approach}
Under Assumption \ref{ass:non-degeneracy}, \ref{asmp:feasibility_uniqueness} and \ref{assu:convexity}, for a non-degenerate scenario problem $\text{SP}_{\mathcal{N}}$, it holds that
\begin{equation}
\label{eqn:exact_distribution_violation_prob}
  \mathbb{P}^N \Big( \mathbb{V}(x_N^\ast)  > \epsilon \Big) \le \sum_{i=1}^{n-1} \binom{N}{i} \epsilon^i (1- \epsilon )^{N-i}.
\end{equation}
The probability $\mathbb{P}^N$ is taken with respect to $N$ random scenarios $\mathcal{N} = \{\xi^i\}_{i=1}^N$.
\end{thm}
A scenario problem $\text{SP}_{\mathcal{N}}$ is \emph{fully-supported} if the number of support scenarios equates the number of decision variables, i.e. $|\mathcal{S}| = n$. The inequality (\ref{eqn:exact_distribution_violation_prob}) is tight for fully-support problems. For non-fully supported problems, if the number of support scenarios is bounded by a known value $h$, i.e. $|\mathcal{S}| \le h < n$, then \cite{calafiore_random_2010} shows that (\ref{eqn:exact_distribution_violation_prob}) could be tightened as
\begin{equation}
\label{eqn:exact_distribution_violation_prob_tightened}
  \mathbb{P}^N \Big( \mathbb{V}(x_N^\ast)  > \epsilon \Big) \le \sum_{i=1}^{h-1} \binom{N}{i} \epsilon^i (1- \epsilon )^{N-i}.
\end{equation}
Based on Theorem \ref{thm:exact_feasibility_scenario_approach}, the scenario approach answers the sample complexity question in Corollary \ref{cor:prior_sample_complexity_fully_supported}.
\begin{cor}[\cite{campi_exact_2008,calafiore_random_2010}]
\label{cor:prior_sample_complexity_fully_supported}
Given a violation probability $\epsilon \in (0,1)$ and a confidence parameter $\beta \in (0,1)$, if we choose the smallest number of scenarios $N$ such that 
\begin{equation}
  \sum_{i=0}^{h-1} \binom{N}{i} \epsilon^i (1- \epsilon)^{N-i} \le \beta,
\end{equation}
then it holds that
\begin{equation}
  \mathbb{P}^N\Big ( \mathbb{V}(x_{\mathcal{N}}^\ast) \le \epsilon \Big) \ge 1 - \beta,
\end{equation}  
where $x_{\mathcal{N}}^\ast$ is the optimal solution to $\text{SP}_{\mathcal{N}}$, and $h$ is the upper bound on the number of support scenarios, i.e. $|\mathcal{S}| \le h \le n$. 
\end{cor}
The scenario approach is essentially a randomized algorithm to solve the chance-constrained optimization problem (\ref{form:cc_opt}). The randomness of the scenario approach comes from drawing i.i.d. scenarios. The confidence parameter $\beta$ quantifies the risk of failure due to drawing scenarios from a ``bad'' set. Corollary \ref{cor:prior_sample_complexity_fully_supported} shows that by choosing a proper number of scenarios, the corresponding optimal solution $x_{\mathcal{N}}^*$ will have violation probability less than $\epsilon$ with high confidence $1 - \beta$.

The scenario approach is a very simple yet powerful method. It is particularly attractive due to the \emph{distribution-free} feature. Theorem \ref{thm:exact_feasibility_scenario_approach} (and Corollary \ref{cor:prior_sample_complexity_fully_supported}) holds for any types of distribution. It requires nothing except the i.i.d. drawing of scenarios. 
We further explore the strength of the scenario approach in this paper. In addition to the distribution-free feature, we show that the scenario approach can go beyond the convexity assumption and be applied on non-convex problems in certain circumstances.
\begin{rem}[The Role of Convexity]
\label{rem:convexity}
Most results of the scenario approach, e.g. \cite{campi_exact_2008,calafiore_random_2010}, are built upon the convexity assumption (i.e. Assumption \ref{assu:convexity}). It plays a major role in bounding the number of support scenarios. Because of the convexity assumption \footnote{This originates from the Helly's lemma in convex analysis.}, the number of support scenarios is bounded by the number of decision variables $n$. 
For non-convex problems, the number of support scenarios could be more than $n$, e.g. \cite{campi_general_2018}.
After carefully examining the proofs of Theorem \ref{thm:exact_feasibility_scenario_approach} and Corollary \ref{cor:prior_sample_complexity_fully_supported} in \cite{campi_exact_2008} and \cite{calafiore_random_2010}, however, we would like to point out that bounding the number of support scenarios is indeed \emph{the only role} of the convexity assumption. The remaining parts of the proofs of Theorem \ref{thm:exact_feasibility_scenario_approach} and Corollary \ref{cor:prior_sample_complexity_fully_supported} do \emph{not} rely on the convexity assumption. In other words, if we are able to find $|\mathcal{S}| \le h$ for some non-convex problems satisfying the non-degeneracy assumption \ref{ass:non-degeneracy} \footnote{For non-convex problems, it is likely that $h > n$, which might lead to a large number of scenarios required by the theory. Fortunately, the unit commitment problem is not the case.}, Theorem \ref{thm:exact_feasibility_scenario_approach} and Corollary \ref{cor:prior_sample_complexity_fully_supported} still hold true despite the non-convexity. 
\end{rem}

\section{Deterministic and Chance-constrained Unit Commitment} 
\label{sec:unit_commitment}
\subsection{Nomenclature} 

\begin{xtabular}{ll}
\multicolumn{2}{l}{Constants and Parameters} \\
$a^k \in \{0,1\}^{n_g}$&  generator availability in contingency $k$ \\
$\alpha_k \in \mathbf{R}_+$ & weight of contingency $k$ \\
$c_g \in \mathbf{R}^{n_g}$ & generation costs \\
$c_z \in \mathbf{R}^{n_g}$ & no load cost \\
$c_r \in \mathbf{R}^{n_g}$ & reserve costs \\
$c_u \in \mathbf{R}^{n_g}, c_v \in \mathbf{R}^{n_g}$ & startup/shutdown cost  \\
$\hat{d}^{t} \in \mathbf{R}^{n_d},\tilde{d}^{t} \in \mathbf{R}^{n_d}$ & load forecast and forecast error (time $t$) \\
$\hat{w}^{t} \in \mathbf{R}^{n_w},\tilde{w}^{t} \in \mathbf{R}^{n_w}$ & wind forecast and forecast error (time $t$) \\
$\underline{g} \in \mathbf{R}^{n_g}, \overline{g} \in \mathbf{R}^{n_g}$ & generation lower and upper bounds \\
$\underline{\gamma} \in \mathbf{R}^{n_g}, \overline{\gamma} \in \mathbf{R}^{n_g}$ & ramping lower and upper bounds \\
$\underline{u}_i \in \mathbf{R}_+, \underline{v}_i \in \mathbf{R}_+$ & minimum on/off time for generator $i$ \\
\multicolumn{2}{l}{Indices} \\
$k \in \{0,1,\cdots,n_k\}$ & contingency index\\
$t \in \{1,2, \cdots,n_t\}$ & time (snapshot) index \\
\multicolumn{2}{l}{Binary Decision Variables (time $t$)} \\
$z^{t} \in \{0,1\}^{n_g}$ & generator on/off states\\
$u^t \in \{0,1\}^{n_g}$ & generator $i$ turned on at $t$ if $u_i^t = 1$ \\
$v^t \in \{0,1\}^{n_g}$ & generator $i$ turned off at $t$ if $v_i^t = 1$ \\
\multicolumn{2}{l}{Continuous Decision Variables (time $t$, contingency $k$)} \\
$g^{t,k} \in \mathbf{R}^{n_g}$ & generation output \\
$r^{t} \in \mathbf{R}^{n_g}$ & reserve 
\end{xtabular}

The number of loads, generators, wind farms, contingencies, and snapshots are denoted by $n_d,n_g,n_w,n_k$ and $n_t$, respectively.
\subsection{Deterministic Unit Commitment} 
\label{sub:deterministic_unit_commitment}
The \emph{deterministic} Unit Commitment (d-UC) problem \ref{opt:det-SCUC} seeks optimal commitment and startup/shutdown decisions $(z^t,u^t, v^t)$, generation and reserve schedules $(g^{t,k}, r^t)$. 
\begin{subequations}
\label{opt:det-SCUC}
\begin{align}
  \min_{z,u,v,g,r}~ & \sum_{t=1}^{n_t} \Big(  c_z^\intercal z^t + c_u^\intercal u^t + c_v^\intercal v^t + c_r^\intercal r^t + \sum_{k=0}^{n_k} \alpha_k c_g^\intercal g^{t,k} \Big) \label{opt:det-SCUC-obj} \\
\text{s.t.}~& \mathbf{1}^\intercal g^{t,k} + \mathbf{1}^\intercal \hat{w}^t  \ge \mathbf{1}^\intercal \hat{d}^t \label{opt:det-SCUC-balance}\\
& a^k \circ \underline{\gamma} \le g^{t,k} - g^{t-1,k} \le a^k \circ \overline{\gamma}  \label{opt:det-SCUC-ramp} \\
& a^k \circ (g^{t,0} - r^t) \le g^{t,k} \le a^k \circ (g^{t,0} + r^t) \label{opt:det-SCUC-contingency} \\
& a^k \circ  \underline{g} \circ  z^t \le g^{t,k} \le a^k \circ  \overline{g} \circ  z^t \label{opt:det-SCUC-contgencap-redundant} \\
& \hspace{109pt} k \in [0,n_k], t \in [1,n_t] \nonumber \\
& \underline{g} \circ  z^t \le g^{t,0} \le \overline{g} \circ  z^t \label{opt:det-SCUC-gencap}\\
& \underline{g} \circ  z^t \le g^{t,0} - r^t \le g^{t,0} + r^t \le \overline{g} \circ  z^t \label{opt:det-SCUC-resgencap}\\
& z^{t-1} - z^t + u^t \ge 0 \label{opt:det-SCUC-startup} \\
& z^t - z^{t-1} + v^t \ge 0 \label{opt:det-SCUC-shutdown} \\
& \hspace{157pt} t \in [1,n_t] \nonumber \\
& z_i^{t} - z_i^{t-1} \le z_i^\iota, ~ i \in [1,n_g] \label{opt:det-SCUC-minon-time} \\
& \hspace{23pt} \iota \in [t+1,\min\{t+\underline{u}_i-1,n_t\}], t \in [2,n_t] \nonumber \\
& z_i^{t-1} - z_i^{t} \le 1- z_i^\iota, ~ i \in [1,n_g] \label{opt:det-SCUC-minoff-time}\\
& \hspace{23pt} \iota \in [t+1,\min\{t+\underline{v}_i-1,n_t\}], t \in [2,n_t] \nonumber
\end{align}
\end{subequations}
The objective of (\ref{opt:det-SCUC}) is to minimize total operation costs, which include 
no-load costs $c_z^\intercal z^t$, startup costs $c_u^\intercal u^t$, shutdown costs $c_v^\intercal v^t$, generation costs $c_g^\intercal g^{t,k}$ and reserve costs $c_r^\intercal s^t$. Security constraints ensure: enough supply to meet demand (\ref{opt:det-SCUC-balance}), generation levels within ramping limits (\ref{opt:det-SCUC-ramp}) and within capacity within limits (\ref{opt:det-SCUC-contgencap-redundant})-(\ref{opt:det-SCUC-gencap}) in any contingency $k$. Constraints (\ref{opt:det-SCUC-contingency}) and (\ref{opt:det-SCUC-resgencap}) are about the relationship between generation and reserve in any contingency $k$. Constraints (\ref{opt:det-SCUC-startup})-(\ref{opt:det-SCUC-shutdown}) are the logistic constraints about commitment status, startup and shutdown decisions. Minimum on/off time constraints for all generators are in (\ref{opt:det-SCUC-minon-time})-(\ref{opt:det-SCUC-minoff-time}).
It is worth mentioning that constraints (\ref{opt:det-SCUC-ramp})-(\ref{opt:det-SCUC-resgencap}) also guarantee the consistency\footnote{If generator $i$ fails in contingency $k$ (i.e. $a_i^k = 0$), then $g_i^{t,k} = 0, \forall t \in [1,n_t]$. Similarly, if generator $i$ is not committed at time $t$ (i.e. $z_i^t = 0$), then $g_i^{t,k} = 0$, $\forall t \in [1,n_t], k \in [0,n_k]$.}
of generation levels $g^{t,k}$ with commitment decisions $z^t$ and generator availability $a^k$ in contingency $k$. 
\begin{rem}
\label{rem:redundant_constraints}
Constraint (\ref{opt:det-SCUC-contgencap-redundant}) is redundant, since it is implied
by constraints (\ref{opt:det-SCUC-contingency}), (\ref{opt:det-SCUC-gencap}) and (\ref{opt:det-SCUC-resgencap}). When $z_i^t=0$, constraint (\ref{opt:det-SCUC-contgencap-redundant}) requires $g_i^{t,k} = 0$, which is implied by (\ref{opt:det-SCUC-contingency}) and (\ref{opt:det-SCUC-gencap}). Similarly, when generator $i$ is not available in contingency $k$ ($a_i^k = 0$),  (\ref{opt:det-SCUC-contingency}) implies constraint (\ref{opt:det-SCUC-contgencap-redundant}) $g_i^{t,k} = 0$. In other cases, constraint (\ref{opt:det-SCUC-contgencap-redundant}) for generator $i$ is equivalent with $\underline{g}_i \le g_i^{t,k} \le \overline{g}_i$, which can be derived from (\ref{opt:det-SCUC-contingency}) and (\ref{opt:det-SCUC-resgencap}). Constraint (\ref{opt:det-SCUC-contgencap-redundant}) is omitted in the remainder of this paper.
\end{rem}


\subsection{Chance-constrained Unit Commitment} 
\label{sub:chance_constrained_unit_commitment}
The deterministic Unit Commitment formulation utilizes the expected wind generation and load forecast, it does not take the uncertainties from wind and load into consideration. We propose an improved formulation of d-UC using chance constraints, which guarantee the system security with a tunable level of risk $\epsilon$ with respect to uncertainties.
\begin{subequations}
\label{opt:cc-SCUC}
\begin{align}
  \min_{z,u,v,g,r}~ &  (\ref{opt:det-SCUC-obj}) \nonumber \\
\text{s.t.}~& (\ref{opt:det-SCUC-balance})(\ref{opt:det-SCUC-ramp})(\ref{opt:det-SCUC-contingency})(\ref{opt:det-SCUC-gencap})(\ref{opt:det-SCUC-resgencap})(\ref{opt:det-SCUC-startup})(\ref{opt:det-SCUC-shutdown})(\ref{opt:det-SCUC-minon-time})(\ref{opt:det-SCUC-minoff-time}) \nonumber \\
& \mathbb{P}_{\tilde{w}\times\tilde{d}}\Big( \mathbf{1}^\intercal g^{t,k} + \mathbf{1}^\intercal (\hat{w}^t + \tilde{w}^t) \ge \mathbf{1}^\intercal (\hat{d}^t + \tilde{d}^t ), \nonumber \\
& \hspace{1.5cm} k \in [0,n_k], t \in [1,n_t] \Big ) \ge 1 - \epsilon \label{opt:cc-SCUC-balance-U} 
\end{align}
\end{subequations}
Problem (\ref{opt:cc-SCUC}) is the formulation of chance-constrained Unit Commitment (c-UC). Instead of using expected load $\hat{d}^t$ as in (\ref{opt:det-SCUC}), we consider loads $d^t$ as forecast $\hat{d}^t$ plus a random forecast error $\tilde{d}^t$ (i.e. $d^t = \hat{d}^t + \tilde{d}^t$).

Comparing with d-UC, the only difference of c-UC is the addition of the chance constraint (\ref{opt:cc-SCUC-balance-U}). The chance constraint guarantees there will be enough supply to meet the net demand in any contingency case at any time 
\begin{equation}
\label{opt:inner-constraint-chance}
\mathbf{1}^\intercal g^{t,k} + \mathbf{1}^\intercal (\hat{w}^t + \tilde{w}^t) \ge \mathbf{1}^\intercal (\hat{d}^t + \tilde{d}^t ), k \in [0,n_k], t \in [1,n_t] 
\end{equation}
with probability no less than $1 - \epsilon$.

To reveal the structures of c-UC, we define the sets below:
\begin{subequations}
\begin{align}
\mathcal{B} &:= \big\{ (z,u,v): (\ref{opt:det-SCUC-startup}), (\ref{opt:det-SCUC-shutdown}), (\ref{opt:det-SCUC-minon-time}), (\ref{opt:det-SCUC-minoff-time}) \big\}  \\
\mathcal{C} &:= \big\{ (g,r):  (\ref{opt:det-SCUC-balance}), (\ref{opt:det-SCUC-ramp}), (\ref{opt:det-SCUC-contingency}) \big\}  \\
\mathcal{H} &:= \big\{ (z,g,r): (\ref{opt:det-SCUC-gencap}), (\ref{opt:det-SCUC-resgencap})  \big\} \\
\mathcal{U} &:= \big\{(g): (\ref{opt:inner-constraint-chance})  \big\} 
\end{align}
\end{subequations}
Then c-UC can be succinctly represented as:
\begin{subequations}
\label{opt:c-UC-abstract}
\begin{align}
\min_{z,u,v,g,r}~ & (\ref{opt:det-SCUC-obj}) \nonumber \\
\text{s.t.}~ &  (z,u,v) \in \mathcal{B} \label{opt:c-UC-abstract-B} \\
& (g,r) \in \mathcal{C} \label{opt:c-UC-abstract-C} \\
& (z,g,r) \in \mathcal{H} \label{opt:c-UC-abstract-H} \\
& \mathbb{P}\big ( g \in \mathcal{U}  \big) \ge 1 - \epsilon \label{opt:c-UC-abstract-U}
\end{align}
\end{subequations}
Sets $\mathcal{B}$ and $\mathcal{C}$ stand for the \emph{deterministic} constraints for binary and continuous variables, respectively. 
Set $\mathcal{H}$ represents the hybrid constraints related with both continuous and binary variables. Set $\mathcal{U}$ represents all constraints related with uncertainties.
\begin{rem}
\label{rem:structural_prop_c-UC}
The non-convexity of unit commitment comes from binary variables $(z,u,v)$. Clearly as shown in (\ref{opt:c-UC-abstract}), non-convexity (i.e. set $\mathcal{B}$ and $\mathcal{H}$) only exists in deterministic constraints, and uncertain constraints $\mathcal{U}$ are only related with continuous variables. This observation plays a critical role in analyzing the structural properties of s-UC in Lemma \ref{lem:s_UC_t=1} and Corollary \ref{cor:s_UC_num_support_scenario}.
\end{rem}

\section{Solving c-UC via the Scenario Approach} 
\label{sec:scenario_based_unit_commitment}
\subsection{Scenario-based Unit Commitment} 
\label{sub:scenario_based_uc}
As explained in Section \ref{sub:the_scenario_approach}, the scenario approach reformulates (\ref{form:cc_opt}) to a scenario problem (\ref{form:scenario_problem_set}) using $N$ scenarios. 
For the unit commitment problem, we denote the set of $N$ scenarios as $\mathcal{N} = \{ (\tilde{d}^1,\tilde{w}^1), (\tilde{d}^2,\tilde{w}^2),\cdots, (\tilde{d}^N,\tilde{w}^N) \}$. Each load and wind scenario is a time series of length $n_t$: $\tilde{d}^i= (\tilde{d}^{1,i},\cdots,\tilde{d}^{n_t,i}), ~\tilde{w}^i= (\tilde{w}^{1,i},\cdots,\tilde{w}^{n_t,i})$.
Then we define the set $\mathcal{U}_i$ corresponding to scenario $i$:
\begin{align}
\mathcal{U}_i & := \big \{g: \mathbf{1}^\intercal g^{t,k} + \mathbf{1}^\intercal (\hat{w}^t + \tilde{w}^{t,i}) \nonumber \\
& \hspace{1.7cm} \ge \mathbf{1}^\intercal (\hat{d}^t + \tilde{d}^{t,i}), t \in [1,n_t],  k \in [0,n_k] \big \} \label{eqn:scenario_set_U_i}
\end{align}
The scenario problem for c-UC can be written as
\begin{subequations}
\label{opt:s-UC}
\begin{align}
\min_{z,u,v,g,r}~ & (\ref{opt:det-SCUC-obj}) \nonumber \\
\text{s.t.}~ & (\ref{opt:c-UC-abstract-B}), (\ref{opt:c-UC-abstract-C}), (\ref{opt:c-UC-abstract-H}) \nonumber \\
& g \in \cap_{i=1}^{N} \mathcal{U}_i \label{opt:s-UC-scenario}
\end{align}
\end{subequations}
Problem (\ref{opt:s-UC}) is referred as s-UC in the remainder of this paper.

\subsection{Structural Properties of s-UC} 
\label{sub:structural_properties_of_s_uc}
For notation simplicity, we define $\iota^t$ as the index of the scenario with the largest \emph{net} demand forecast error at time $t$:
\begin{eqnarray}
  \iota^t := \arg_{i} \max\Big\{\mathbf{1}^\intercal \tilde{d}^{t,1}-\mathbf{1}^\intercal\tilde{w}^{t,1},\cdots, \mathbf{1}^\intercal \tilde{d}^{t,N}-\mathbf{1}^\intercal\tilde{w}^{t,N} \Big\}, \label{eqn:support_scenario_index_t}
\end{eqnarray}
and define $\overline{\mathcal{S}} := \{\iota^1,\iota^2,\cdots,\iota^{n_t}\}$. Clearly there might be repetitive scenario indices in $\iota^1,\iota^2,\cdots,\iota^{n_t}$, i.e. $|\overline{\mathcal{S}}| \le n_t$.
\begin{lem}
\label{lem:s_UC_t=1}
When $n_t = 1$, s-UC has at most one support scenario. The support scenario is the one with the largest \emph{net} demand forecast error, i.e. $(\tilde{d}^{1,\iota^1},\tilde{w}^{1,\iota^1})$ if the number of support scenarios is not zero.
\end{lem}
\begin{proof}
Let $\iota^1$ be the scenario index defined in (\ref{eqn:support_scenario_index_t}), clearly $\mathcal{U}_{\iota^1} = \cap_{i=1}^{N} \mathcal{U}_i$, which implies that the removal of any scenario other than $\iota^1$ will not change the feasible region.
According to Definition \ref{defn:support_scenario}, all other scenarios except $\iota^1$ cannot be a support scenario. Therefore s-UC with $n_t=1$ has at most one support scenario.
\end{proof}
\begin{cor}
\label{cor:s_UC_num_support_scenario}
For s-UC (\ref{opt:s-UC}), let $\mathcal{S}$ denote the set of its support scenarios, then $\mathcal{S} \subseteq \overline{\mathcal{S}}$, which indicates $|\mathcal{S}| \le |\overline{\mathcal{S}}| \le n_t$.
\end{cor}
\begin{proof}
Let $\mathcal{U}_i^t = \big \{g^t: \mathbf{1}^\intercal g^{t,k} + \mathbf{1}^\intercal (\hat{w}^t + \tilde{w}^{t,i}) \ge \mathbf{1}^\intercal (\hat{d}^t + \tilde{d}^{t,i}),  k \in [0,n_k] \big \}$, then $\mathcal{U}_i = \mathcal{U}_i^1 \times \mathcal{U}_i^2 \cdots \times \mathcal{U}_i^{n_t}$. According to Lemma \ref{lem:s_UC_t=1}, $\cap_{i=1}^N \mathcal{U}_i^t$ has at most one support scenario, which is indexed by $\iota^{t}$. Applying Lemma \ref{lem:s_UC_t=1} for all $n_t$ snapshots, we can see that the set $\overline{\mathcal{S}}$ contains all candidates to support scenarios, thus $\mathcal{S} \subseteq \overline{\mathcal{S}}$ and $|\mathcal{S}| \le |\overline{\mathcal{S}}| \le n_t$.
\end{proof}

The intuition behind Lemma \ref{lem:s_UC_t=1} and Corollary \ref{cor:s_UC_num_support_scenario} is illustrated in Fig. \ref{fig:case3-degenerate-example-UC}. Fig. \ref{fig:case3-degenerate-example-UC} visualizes the constraints and feasible region $(g_1, g_2)$ of the 2-generator, 3-bus and 3-line system in \cite{geng_general_2019}. Four blue regions ($\mathcal{B}_0,\mathcal{B}_1,\mathcal{B}_2,\mathcal{B}_3$) stand for four possible on/off states of 2 generators. For example, $\mathcal{B}_2$ shows the case in which generator 1 is off ($z_1 = 0$) and generator 2 is on ($z_2 = 1$). The black solid lines represent the determine constraints $\mathcal{C}$. Three dashed/dotted lines denote three constraints ($\mathcal{U}_1,\mathcal{U}_2,\mathcal{U}_3$) of three scenarios. Since the scenario constraint \eqref{eqn:scenario_set_U_i} is only about supply and demand (transmission limits are not included), the feasible region of s-UC is clearly defined by the scenario with the \emph{largest net demand} ($\mathcal{U}_1$ in Fig. \ref{fig:case3-degenerate-example-UC}), which is the support scenario of s-UC. The scenario with the largest net demand at each snapshot is a candidate for support scenarios (Lemma \ref{lem:s_UC_t=1}), therefore there are at most $n_t$ candidates for support scenarios (Corollary \ref{cor:s_UC_num_support_scenario}).  

\begin{figure*}[htbp]
  \centering
  \subfloat[Constraints $\mathcal{B}_0,\mathcal{B}_1,\mathcal{B}_2,\mathcal{B}_3$ representing 4 possible states of 2 generators.]{\includegraphics[width=0.5\linewidth]{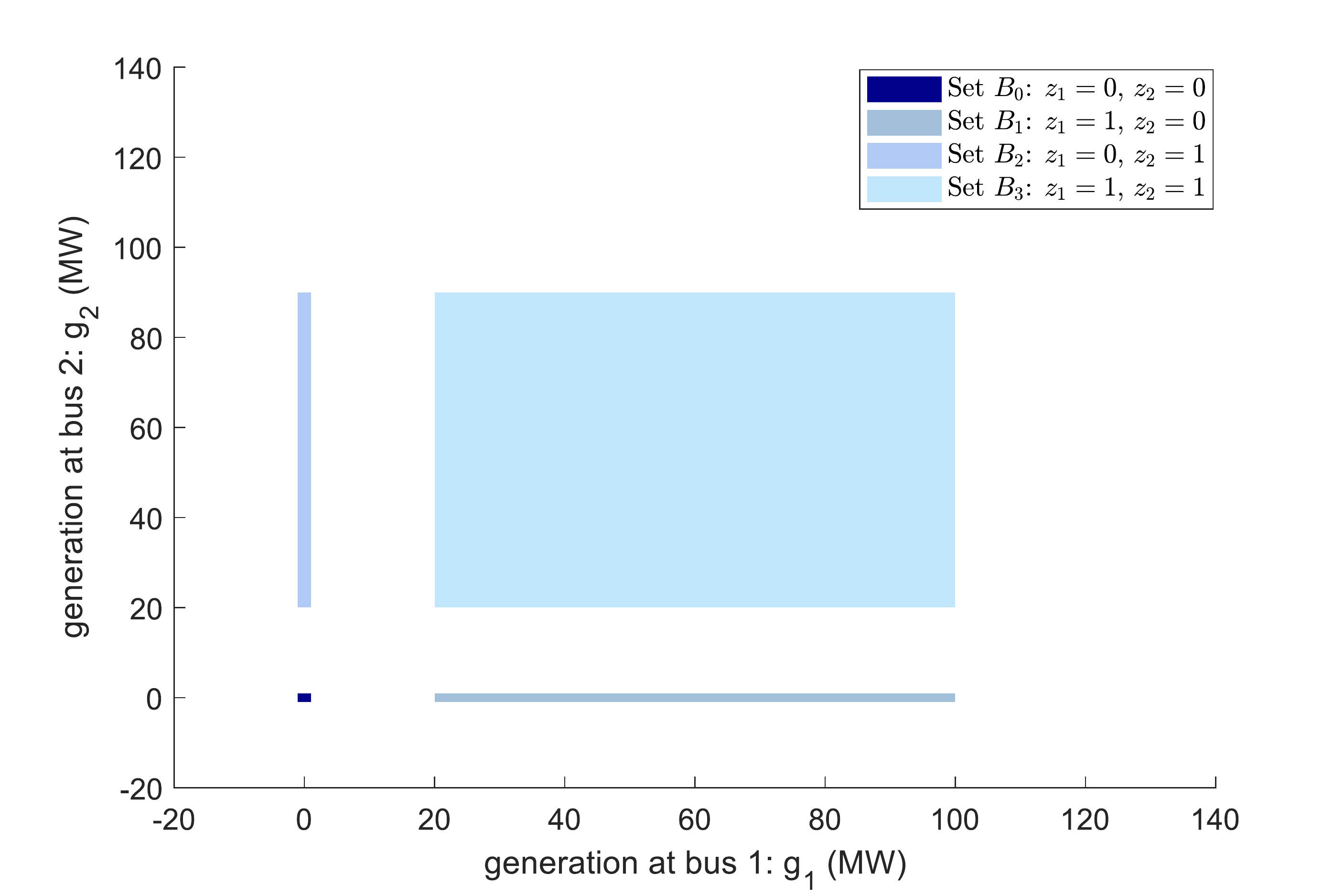}\label{fig:case3-degenerate-example-no}} 
  \subfloat[Deterministic constraints.]{\includegraphics[width=0.5\linewidth]{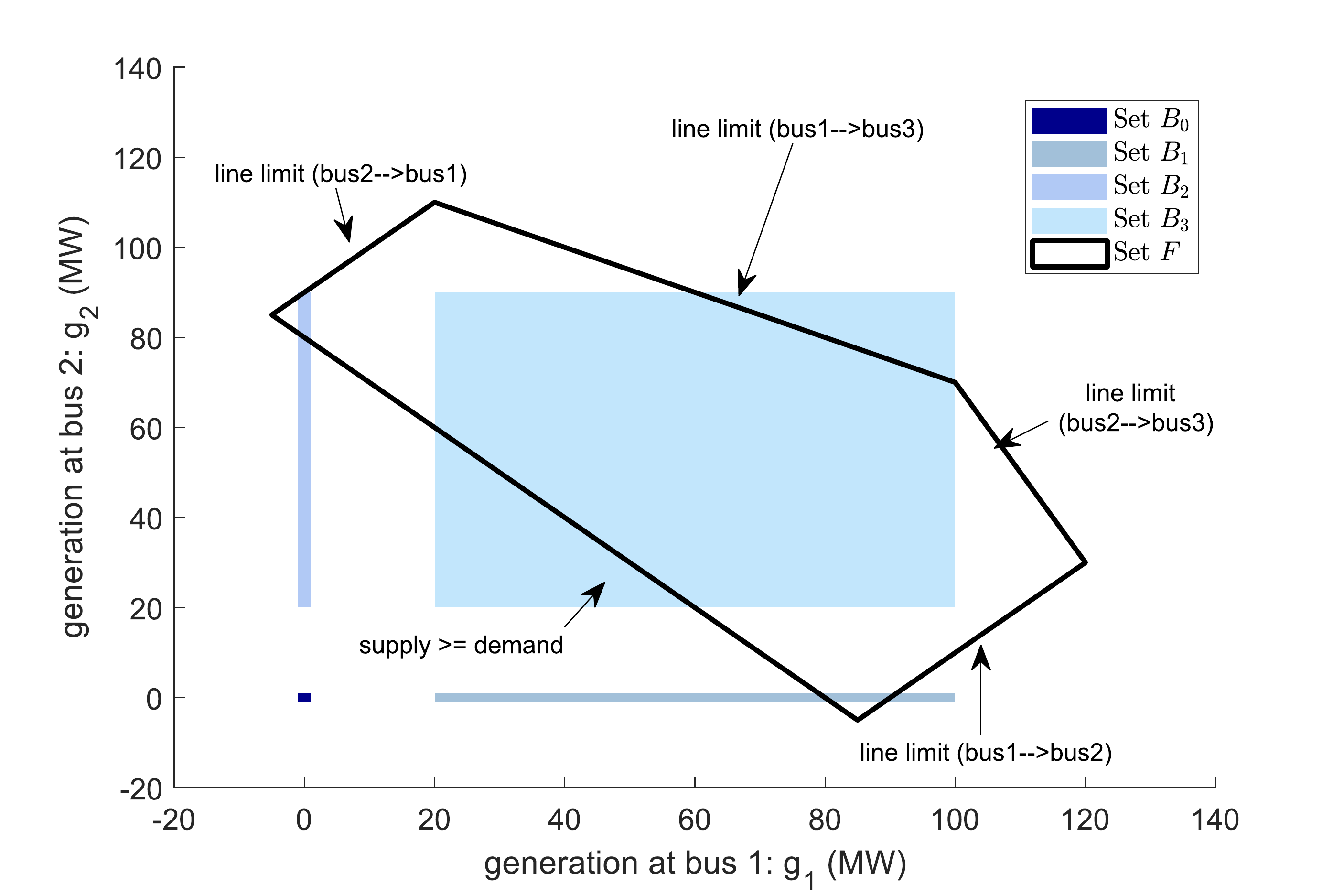}\label{fig:case3-degenerate-example-det}} \\
  \subfloat[Scenario constraints $\mathcal{U}$, with all scenarios ($\mathcal{U}_1$,$\mathcal{U}_2$,$\mathcal{U}_3$).]{\includegraphics[width=0.5\linewidth]{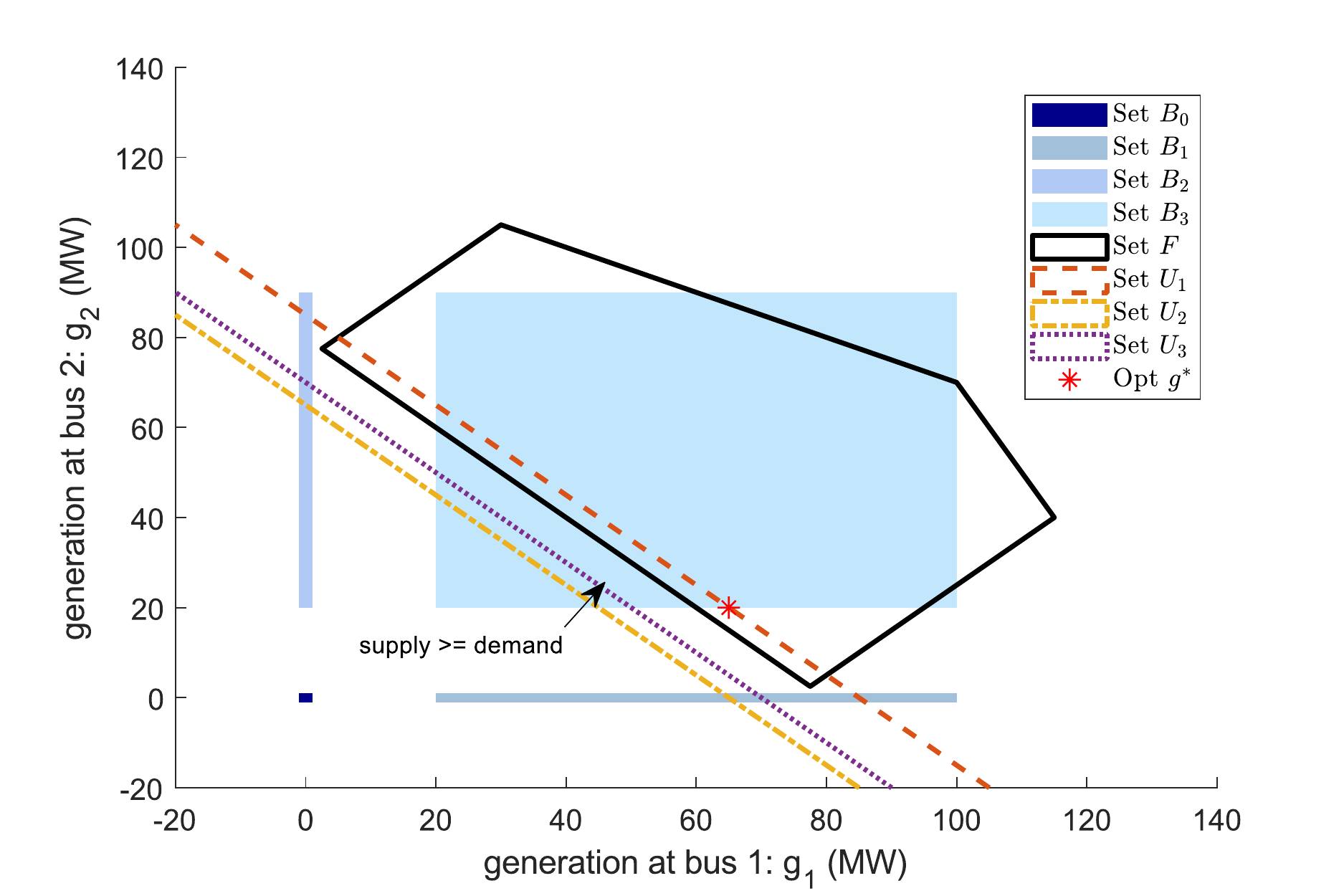}\label{fig:case3-degenerate-example-all-UC}} 
  \subfloat[Scenario constraints $\mathcal{U}$, with only support scenario ($\mathcal{U}_1$).]{\includegraphics[width=0.5\linewidth]{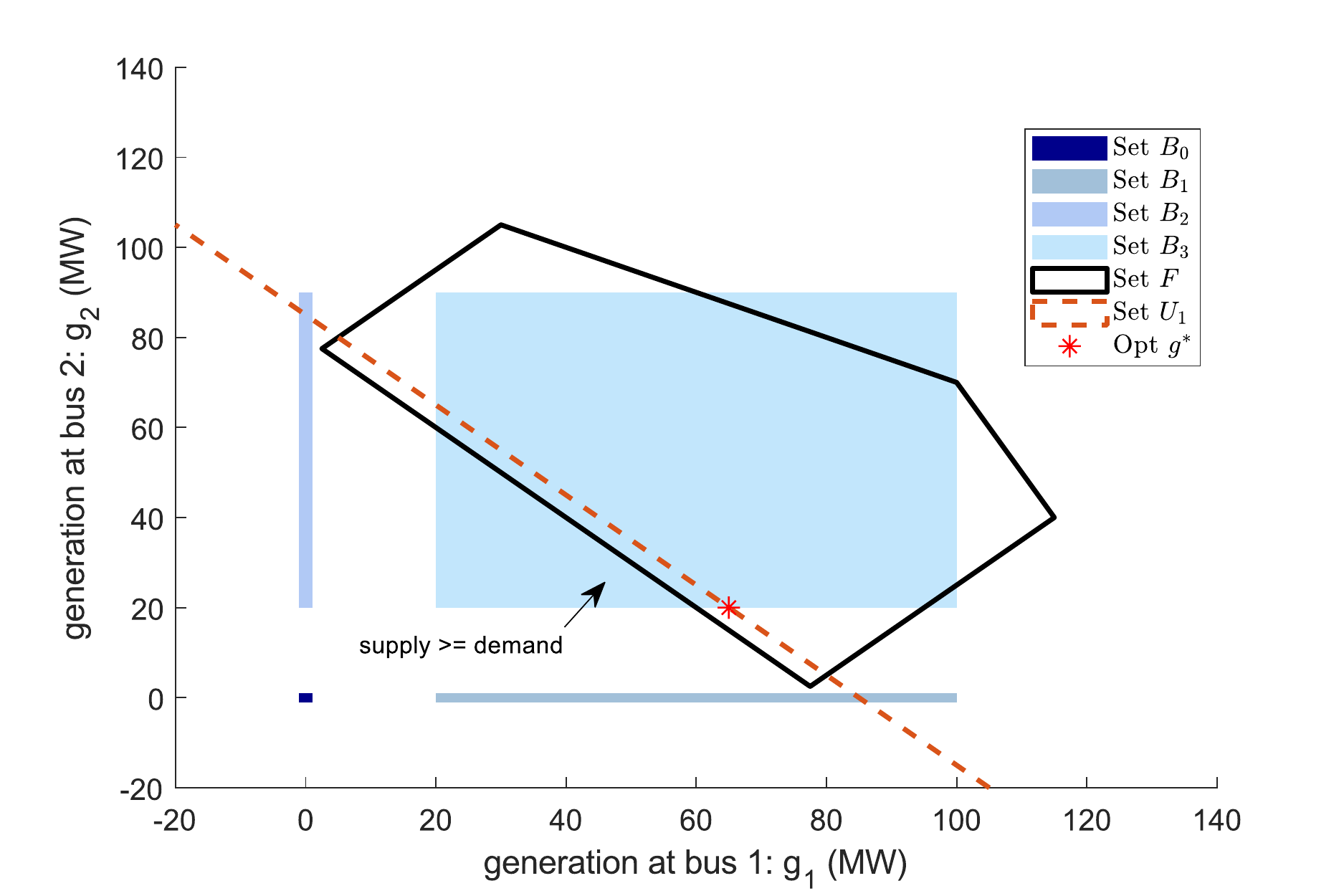}\label{fig:case3-degenerate-example-all-UC}} 
  \caption{Illustration of Lemma \ref{lem:s_UC_t=1} and Corollary \ref{cor:s_UC_num_support_scenario} using the 2-generator, 3-bus system in \cite{geng_general_2019}.}
  \label{fig:case3-degenerate-example-UC}
\end{figure*}

\subsection{Sample Complexity for s-UC} 
\label{sub:sample_complexity_for_s_uc}
Corollary \ref{cor:s_UC_num_support_scenario} shows that $|\mathcal{S}| \le n_t$ for s-UC, then we can use the results in Corollary \ref{cor:prior_sample_complexity_fully_supported} and Remark \ref{rem:convexity} to calculate the number of scenarios to achieve the desired security level $1-\epsilon$ with confidence $1-\beta$. Table \ref{tab:samplex_complexity_case118} presents the sample complexity (number of scenarios) needed with various $\epsilon$ levels for the 118-bus system in Section \ref{sub:settings_of_the_118_bus_system}. 

Although unit commitment is non-convex because of the binary variables $(z,u,v)$. It is in general difficult to estimate the number of support scenarios $|\mathcal{S}|$ a-priori. Without exploiting the structural properties of s-UC as in Corollary \ref{cor:s_UC_num_support_scenario}, the best bound \footnote{Before revealing the structure of s-UC in Corollary \ref{cor:s_UC_num_support_scenario}, $n$ is not an upper bound on $|\mathcal{S}|$ because s-UC is non-convex. But $n$ is the best bound we could hope for using the results in Theorem \ref{thm:exact_feasibility_scenario_approach} and Corollary \ref{cor:prior_sample_complexity_fully_supported}.} might be the number of decision variables $|\mathcal{S}| \le n$, which is $4n_g n_t + n_gn_tn_k = 75168$ for the 118-bus system .
Table \ref{tab:samplex_complexity_case118} also presents the sample complexity using $|\mathcal{S}| \le 75168$. As shown in Table \ref{tab:samplex_complexity_case118}, Corollary \ref{cor:s_UC_num_support_scenario} \emph{greatly} reduces the number of scenarios from some astronomical numbers in the case of $|\mathcal{S}| \le 75168$. Another attractive observation is that the results in Corollary \ref{cor:s_UC_num_support_scenario} holds regardless of the system size.

\begin{table*}[htbp]
  \caption{Sample Complexity for s-UC (Case118, $\beta = 10^{-4}$)}
  \label{tab:samplex_complexity_case118}
  \centering

  \begin{tabular}{l|ccccccccc}
  \hline

  \hline
  Violation probability $\epsilon$ & 0.3 & 0.2 & 0.1 & 0.075 & 0.05 & 0.025 & 0.01 \\
  \hline
  Sample Complexity $N$ (when $|\mathcal{S}| \le 24$) & 143 & 221 &   455     &    610    &     921     &   1853        & 4650     \\
  Sample Complexity $N$ (when $|\mathcal{S}| \le 75168$)  & 253416 & 380419 & 761394 & 1015370 & 1523320 & 3047161 & 7618678 \\
  \hline

  \hline
  \end{tabular}
\end{table*}

\section{Case Study} 
\label{sec:case_study}
\subsection{Settings of the 118-bus System} 
\label{sub:settings_of_the_118_bus_system}
We solve the unit commitment problem of an 118-bus system with $54$ generators ($n_g=54$) in 24 hours ($n_t =24$) under 54 possible generator failure contingencies ($n_k=54$). The test system is a modified version of the 118-bus system in \cite{pena_extended_2018}. The modified 118-bus system includes 5 wind farms at different locations.





The numerical simulation was conducted on a desktop with Intel Core i7-2600 CPU@3.40GHz and 16GB of memory.
Matpower and YALMIP were used to formulate the c-UC problem in MatlabR2018a. The c-UC problem was converted to s-UC via ConvertChanceConstraint in \cite{geng_data-driven_2019-2}, then solved using Gurobi 8.10 till the MIP gap is smaller than $0.01\%$.

\subsection{Numerical Results} 
\label{sub:numerical_results}
We solve the s-UC problem with different number of scenarios $N$. Given $N$, we conduct 10 independent Monte-Carlo simulations to examine the randomness of the scenario approach. Another independent test dataset of $10^4$ points was used to evaluate the out-of-sample violation probability $\epsilon$ of the solution to s-UC.

Figure \ref{fig:sc-UC-case118iit-fig_objeps_2axis} demonstrates the optimal objective values and out-of-sample $\epsilon$ with different number of scenarios. As the scenario approach theory suggests, with an increasing number of scenarios, the system risk level $\epsilon$ decreases. Figure \ref{fig:sc-UC-case118iit-fig_objeps_2axis} also shows that with $0.96\%$ of cost increase (from $1.356 \times 10^{6}$ to $1.369 \times 10^{6}$), the system risk $\epsilon$ is reduced from $19\%$ to $2\%$. 
\begin{figure}[htbp]
  \centering
  \includegraphics[width=\linewidth]{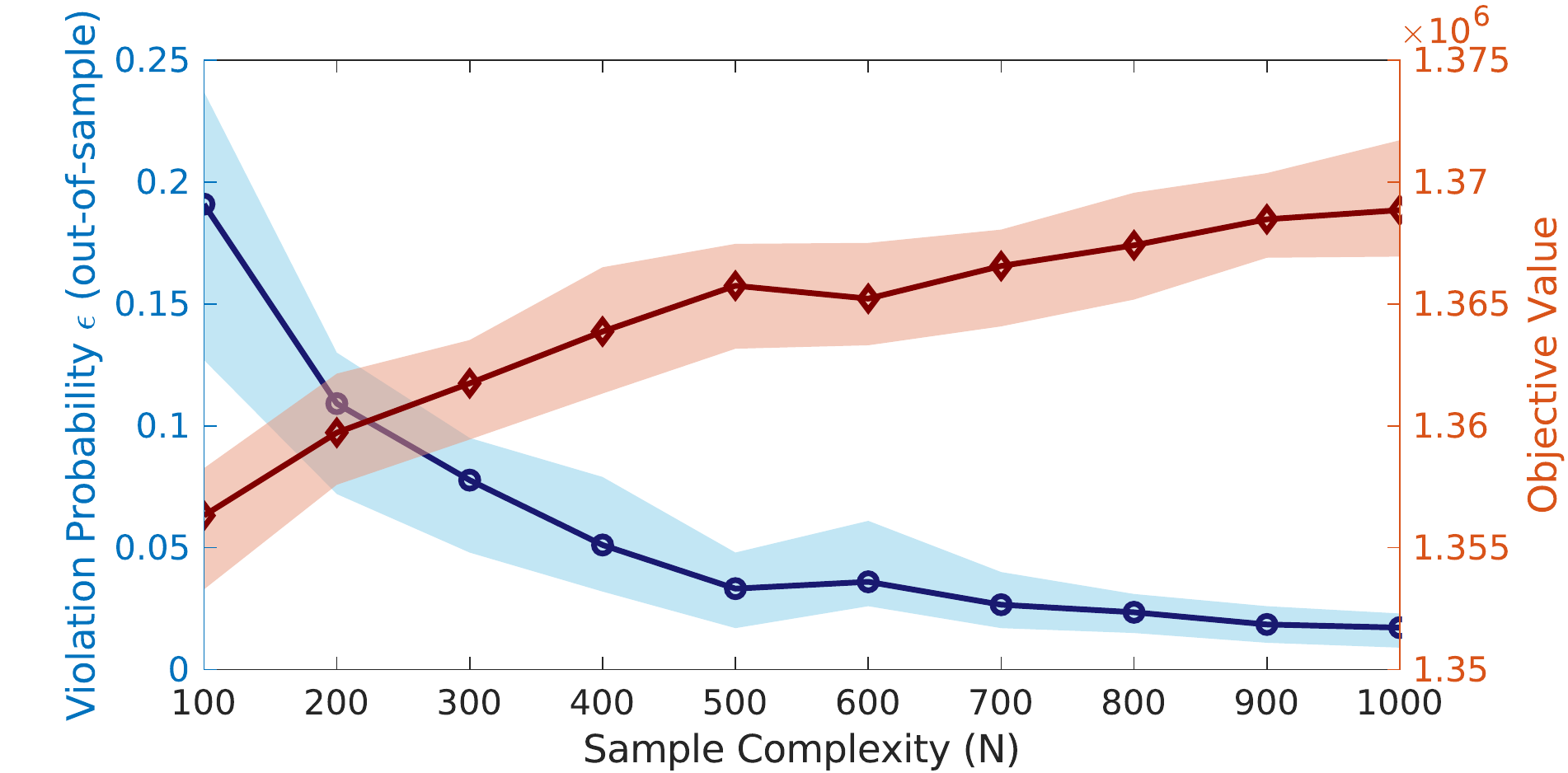}
  \caption{Key Results of s-UC with Different Sample Complexity}
  \label{fig:sc-UC-case118iit-fig_objeps_2axis}
\end{figure}

Figure \ref{fig:sc-UC-case118iit-epsilon-epsilon} plots two violation probabilities. The blue solid curve illustrates the average empirical $\epsilon$ (evaluated on the test dataset of $10^4$ points), the shaded area shows the largest and smallest violation probabilities in 10 Monte-Carlo runs. The dotted green lines plots the guaranteed $\epsilon$ by combining Theorem \ref{thm:exact_feasibility_scenario_approach} with Corollary \ref{cor:s_UC_num_support_scenario}. Figure \ref{fig:sc-UC-case118iit-epsilon-epsilon} shows that the scenario approach is applicable on the unit commitment problem, despite its non-convexity. Furthermore, Figure \ref{fig:sc-UC-case118iit-epsilon-epsilon} also demonstrates the value of Corollary \ref{cor:s_UC_num_support_scenario}. Without showing that $|\mathcal{S}| \le n_t$ as in Corollary \ref{cor:s_UC_num_support_scenario}, Theorem \ref{thm:exact_feasibility_scenario_approach} is only able to provide useless guarantees (e.g. $\epsilon \le 0.999999$ when using $1000$ scenarios). 
\begin{figure}[htbp]
  \centering
  \includegraphics[width=\linewidth]{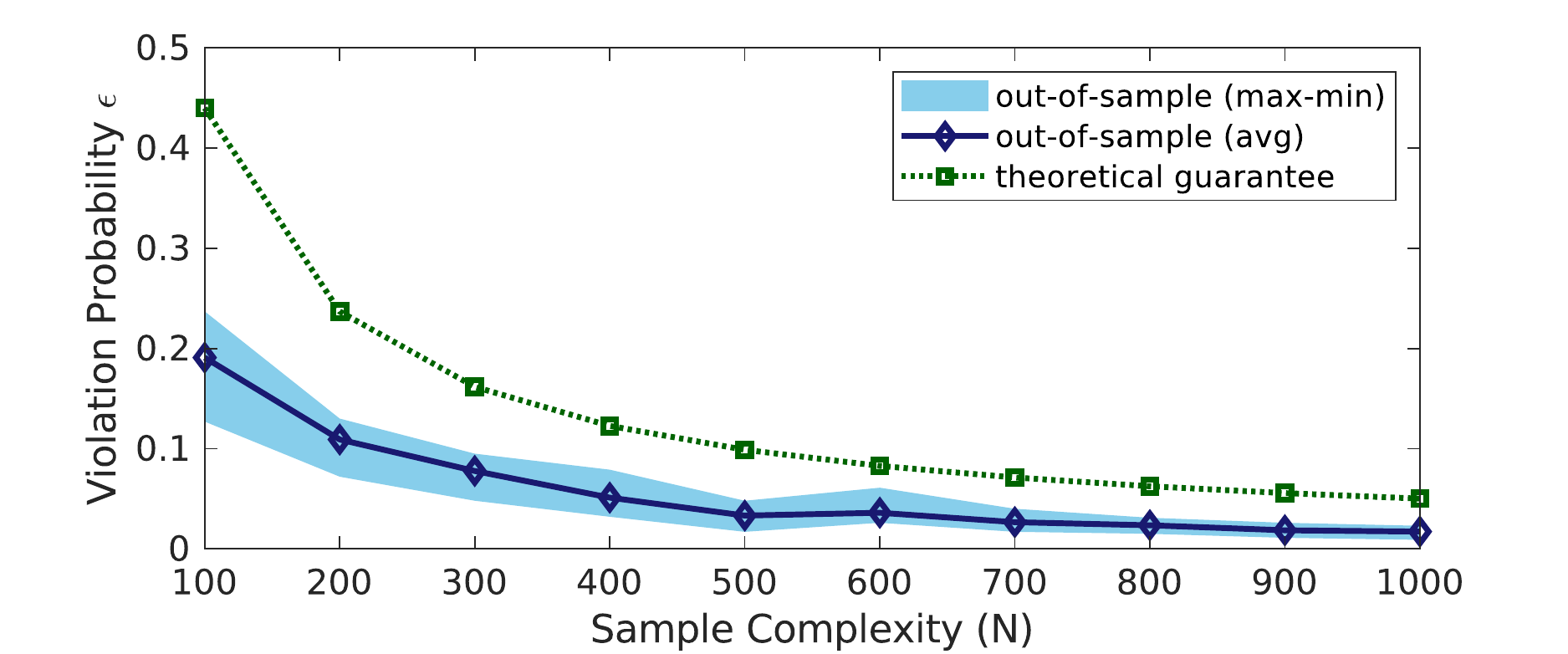}
  \caption{Theoretical and Empirical Violation Probabilities $\epsilon$}
  \label{fig:sc-UC-case118iit-epsilon-epsilon}
\end{figure}

Due to the non-convexity from the binary decision variables, the scenario approach was considered not applicable on the unit commitment problem previously. 
One main contribution of this paper is to show the potential of the scenario approach on non-convex problems like unit commitment. By exploring the structural properties of s-UC, Section \ref{sec:scenario_based_unit_commitment} shows that the scenario approach could still provide rigorous guarantees on the quality of solutions, as in the convex case. This is all based on Lemma \ref{lem:s_UC_t=1} and Corollary \ref{cor:s_UC_num_support_scenario}. Table \ref{tab:number_support_scenarios} shows the maximum and minimum number of support scenarios in 10 Monte Carlo runs of each given sample complexity $N$. This verifies the correctness of Corollary \ref{cor:s_UC_num_support_scenario}.

\begin{table*}[tb]
  \caption{Number of Support Scenarios}
  \label{tab:number_support_scenarios}
  \centering

  \begin{tabular}{l|cccccccccc}
  \hline

  \hline
  $N$ & 100 & 200 & 300 & 400 & 500 & 600 & 700 & 800 & 900 & 1000 \\
  \hline
  $|\mathcal{S}|$ (min) & 19 & 21 & 22 & 22 & 22 & 22 & 23 & 23 & 23 & 24  \\
  $|\mathcal{S}|$ (max) & 24 & 24 & 24 & 24 & 24 & 24 & 24 & 24 & 24 & 24 \\
  \hline

  \hline
  \end{tabular}
\end{table*}


\subsection{Scenario Reduction} 
\label{sub:scenario_reduction}
When the desired risk level $\epsilon$ is very small, the scenario approach might require a large number of scenarios. This will directly cause memory and computation issues in numerical simulations. Corollary \ref{cor:s_UC_num_support_scenario} turns out to be quite helpful in improving the computational performance. Corollary \ref{cor:s_UC_num_support_scenario} shows that a majority of the scenarios have no impacts on the final solution and thus can be reduced. Then s-UC only needs to be solved with at most $n_t=24$ scenarios, which can be easily identified as mentioned in Section \ref{sub:structural_properties_of_s_uc}. 
We compare the results of using 1000 scenarios with those of using identified 24 (out of 1000) scenarios. Although the optimal solution is slightly different due to a few identical generators, the difference in the objective value is less than $10^{-6}$.




\subsection{Adding Security Constraints} 
\label{sub:adding_security_constraints}
The main limitation of this paper is not considering possible security constraints such as transmission line limits. The nice results in Corollary \ref{cor:s_UC_num_support_scenario} holds only in the absence of a transmission network.
We also applied the scenario approach on chance-constrained SCUC. Numerical results show that the number of support scenarios could be more than $n_t = 24$, but this number does not increase too much (e.g. $30\sim 50$ for the 118-bus system with 186 lines). However, we are yet not able to prove nice results as in Corollary \ref{cor:s_UC_num_support_scenario}. This is one critical part of our ongoing works and beyond the scope of this paper.



\section{Concluding Remarks} 
\label{sec:concluding_remarks}
This paper is a first step towards a practical and rigorous day-ahead decision making framework in uncertain environments. We formulate the chance-constrained unit commitment problem and solve it via the scenario approach. We show that the number of support scenarios in the unit commitment problem is at most $n_t$. This structural property makes the scenario approach applicable in the presence of non-convexity. It substantially reduces the necessary number of scenarios and could be further exploited to reduce the computational requirement to solve the problem. Future work will extend the results towards security-constrained unit commitment. 

\bibliographystyle{IEEEtran}
\bibliography{myreferences}

\begin{thebibliography}{10}
\providecommand{\url}[1]{#1}
\csname url@samestyle\endcsname
\providecommand{\newblock}{\relax}
\providecommand{\bibinfo}[2]{#2}
\providecommand{\BIBentrySTDinterwordspacing}{\spaceskip=0pt\relax}
\providecommand{\BIBentryALTinterwordstretchfactor}{4}
\providecommand{\BIBentryALTinterwordspacing}{\spaceskip=\fontdimen2\font plus
\BIBentryALTinterwordstretchfactor\fontdimen3\font minus
  \fontdimen4\font\relax}
\providecommand{\BIBforeignlanguage}[2]{{%
\expandafter\ifx\csname l@#1\endcsname\relax
\typeout{** WARNING: IEEEtran.bst: No hyphenation pattern has been}%
\typeout{** loaded for the language `#1'. Using the pattern for}%
\typeout{** the default language instead.}%
\else
\language=\csname l@#1\endcsname
\fi
#2}}
\providecommand{\BIBdecl}{\relax}
\BIBdecl

\bibitem{takriti_stochastic_1996}
S.~Takriti, J.~R. Birge, and E.~Long, ``A stochastic model for the unit
  commitment problem,'' \emph{IEEE Trans. Power Syst}, 1996.

\bibitem{wu_stochastic_2007}
L.~Wu, M.~Shahidehpour, and T.~Li, ``Stochastic {Security}-{Constrained} {Unit}
  {Commitment},'' \emph{IEEE Transactions on Power Systems}, 2007.

\bibitem{zheng_stochastic_2015}
Q.~P. Zheng, J.~Wang, and A.~L. Liu, ``Stochastic optimization for unit
  commitment—{A} review,'' \emph{IEEE Trans. Power Syst}, 2015.

\bibitem{bertsimas_adaptive_2013}
D.~Bertsimas, E.~Litvinov, X.~A. Sun, J.~Zhao, and T.~Zheng, ``Adaptive robust
  optimization for the security constrained unit commitment problem,''
  \emph{IEEE Trans. Power Syst}, 2013.

\bibitem{geng_data-driven_2019-2}
X.~Geng and L.~Xie, ``Data-driven decision making in power systems with
  probabilistic guarantees: {Theory} and applications of chance-constrained
  optimization,'' \emph{Annual Reviews in Control}, 2019.

\bibitem{ozturk_solution_2004}
U.~A. Ozturk, M.~Mazumdar, and B.~A. Norman, ``A solution to the stochastic
  unit commitment problem using chance constrained programming,'' \emph{IEEE
  Trans. Power Syst}, 2004.

\bibitem{pozo_chance-constrained_2013}
D.~Pozo and J.~Contreras, ``A chance-constrained unit commitment with an
  \$n-k\$ security criterion and significant wind generation,'' \emph{IEEE
  Trans. Power Syst}, 2013.

\bibitem{wang_chance-constrained_2012}
Q.~Wang, Y.~Guan, and J.~Wang, ``A chance-constrained two-stage stochastic
  program for unit commitment with uncertain wind power output,'' \emph{IEEE
  Trans. Power Syst}, 2012.

\bibitem{wang_price-based_2013}
Q.~Wang, J.~Wang, and Y.~Guan, ``Price-based unit commitment with wind power
  utilization constraints,'' \emph{IEEE Trans. Power Syst}, 2013.

\bibitem{zhao_expected_2014}
C.~Zhao, Q.~Wang, J.~Wang, and Y.~Guan, ``Expected value and chance constrained
  stochastic unit commitment ensuring wind power utilization,'' \emph{IEEE
  Trans. Power Syst}, 2014.

\bibitem{wu_chance-constrained_2014}
H.~Wu, M.~Shahidehpour, Z.~Li, and W.~Tian, ``Chance-constrained day-ahead
  scheduling in stochastic power system operation,'' \emph{IEEE Trans. Power
  Syst}, 2014.

\bibitem{tan_hybrid_2016}
W.-S. Tan and M.~Shaaban, ``A hybrid stochastic/deterministic unit commitment
  based on projected disjunctive milp reformulation,'' \emph{IEEE Trans. Power
  Syst}, 2016.

\bibitem{bagheri_data-driven_2017}
A.~Bagheri, C.~Zhao, and Y.~Guo, ``Data-driven chance-constrained stochastic
  unit commitment under wind power uncertainty,'' in \emph{Power \& {Energy}
  {Society} {General} {Meeting}, 2017 {IEEE}}.\hskip 1em plus 0.5em minus
  0.4em\relax IEEE, 2017.

\bibitem{zhang_chance-constrained_2017-1}
Y.~Zhang, J.~Wang, B.~Zeng, and Z.~Hu, ``Chance-{Constrained} {Two}-{Stage}
  {Unit} {Commitment} under {Uncertain} {Load} and {Wind} {Power} {Output}
  {Using} {Bilinear} {Benders} {Decomposition},'' \emph{IEEE Trans on Power
  Systems}, 2017.

\bibitem{jiang_robust_2012}
R.~Jiang, J.~Wang, and Y.~Guan, ``Robust unit commitment with wind power and
  pumped storage hydro,'' \emph{IEEE Trans. Power Syst}, 2012.

\bibitem{margellos_stochastic_2013}
K.~Margellos, V.~Rostampour, M.~Vrakopoulou, M.~Prandini, G.~Andersson, and
  J.~Lygeros, ``Stochastic unit commitment and reserve scheduling: {A}
  tractable formulation with probabilistic certificates,'' in \emph{Control
  {Conference} ({ECC}), 2013 {European}}.\hskip 1em plus 0.5em minus
  0.4em\relax IEEE, 2013.

\bibitem{hreinsson_stochastic_2015}
K.~Hreinsson, M.~Vrakopoulou, and G.~Andersson, ``Stochastic security
  constrained unit commitment and non-spinning reserve allocation with
  performance guarantees,'' \emph{International Journal of Electrical Power \&
  Energy Systems}, 2015.

\bibitem{margellos_road_2014}
K.~Margellos, P.~Goulart, and J.~Lygeros, ``On the road between robust
  optimization and the scenario approach for chance constrained optimization
  problems,'' \emph{IEEE Transactions on Automatic Control}, 2014.

\bibitem{campi_exact_2008}
M.~C. Campi and S.~Garatti, ``The exact feasibility of randomized solutions of
  uncertain convex programs,'' \emph{SIAM Journal on Optimization}, 2008.

\bibitem{campi_scenario_2009}
M.~C. Campi, S.~Garatti, and M.~Prandini, ``The scenario approach for systems
  and control design,'' \emph{Annual Reviews in Control}, 2009.

\bibitem{modarresi_scenario-based_2018}
M.~S. Modarresi, L.~Xie, M.~Campi, S.~Garatti, A.~Caré, A.~Thatte, and
  P.~Kumar, ``Scenario-based {Economic} {Dispatch} with {Tunable} {Risk}
  {Levels} in {High}-renewable {Power} {Systems},'' \emph{IEEE Transactions on
  Power Systems}, 2018.

\bibitem{ming_scenario-based_2017}
H.~Ming, L.~Xie, M.~Campi, S.~Garatti, and P.~Kumar, ``Scenario-based
  {Economic} {Dispatch} with {Uncertain} {Demand} {Response},'' \emph{IEEE
  Transactions on Smart Grid}, 2017.

\bibitem{calafiore_random_2010}
G.~C. Calafiore, ``Random convex programs,'' \emph{SIAM Journal on
  Optimization}, 2010.

\bibitem{campi_general_2018}
M.~C. Campi, S.~Garatti, and F.~A. Ramponi, ``A general scenario theory for
  non-convex optimization and decision making,'' \emph{IEEE Transactions on
  Automatic Control}, 2018.

\bibitem{geng_general_2019}
X.~Geng, L.~Xie, and M.~S. Modarresi, ``A {General} {Scenario} {Theory} for
  {Security}-{Constrained} {Unit} {Commitment} with {Probabilistic}
  {Guarantees},'' \emph{arXiv preprint arXiv:1910.07672}, 2019.

\bibitem{pena_extended_2018}
I.~Peña, C.~B. Martinez-Anido, and B.-M. Hodge, ``An extended {IEEE} 118-bus
  test system with high renewable penetration,'' \emph{IEEE Trans. Power Syst},
  2018.

\end{thebibliography}

\end{document}